\documentclass{sig-alternate-05-2015}

\usepackage{etoolbox}
\makeatletter
\patchcmd{\maketitle}{\@copyrightspace}{}{}{}
\makeatother

\usepackage{amsthm}
\newtheoremstyle{slanted}
{3pt}
{3pt}
{\slshape}
{}
{\bfseries}
{.}
{.5em}
{}
\theoremstyle{slanted}
\newtheorem{thm}{Theorem}[section]

\newtheorem{lem}[thm]{Lemma}
\newtheorem{prop}[thm]{Proposition}

\newcommand{\subparagraph}{}
\usepackage[raggedright]{titlesec}

\begin{document}







\title{Fast Concurrent Cuckoo Kick-out Eviction Schemes for High-Density Tables}
\subtitle{}

\numberofauthors{1} 
%
\author{
%
%
\alignauthor
William Kuszmaul\\
       \affaddr{Stanford University}\\
       \affaddr{450 Serra Mall}\\
       \affaddr{Stanford, CA 94305}\\
       \email{kuszmaul@stanford.edu}
}
\date{5 February 2016}

\maketitle
\begin{abstract}
  
  Cuckoo hashing guarantees constant-time lookups regardless of table
  density, making it a viable candidate for high-density
  tables. Cuckoo hashing insertions perform poorly at high table
  densities, however. In this paper, we mitigate this problem through
  the introduction of novel kick-out eviction
  algorithms. Experimentally, our algorithms reduce the number of bins
  viewed per insertion for high-density tables by as much as a factor
  of ten.

  We also introduce an optimistic concurrency scheme for transactional
  multi-writer cuckoo hash tables (not using hardware transactional
  memory). For delete-light workloads, one of our kick-out algorithms
  avoids all competition between insertions with high probability, and
  significantly reduces transaction-abort frequency. This result is
  extended to arbitrary workloads using a new synchronization
  mechanism called a claim flag.
\end{abstract}

%
%


%
%

%
%



\clearpage
\clearpage

\section{Introduction}

Unlike traditional hashing, Cuckoo hashing maps each key to two
distinct bins using two hash functions.  To insert a key, we simply
look through the bins identified by the two hash functions and insert
into the first containing a free slot. However, if neither bin
contains an empty slot, then we pick one of two bins and
\emph{kick-out} one of the keys it contains. The displaced key, the
\emph{victim}, is then recursively reinserted into the table. The
sequence of kick-outs that results from the initial insert is called a
\emph{kick-out chain}.

Introduced in 2004, Cuckoo hashing guarantees constant-time reads,
overwrites, and deletes \cite{cuckoo04}. In particular, once a key is
inserted into the table, it is guaranteed to be among the two bins to
which it is hashed -- regardless of the table density.  In settings
where memory is expensive this makes Cuckoo hashing especially
valuable. See, for example, the flash-based key-value store system
FlashStore \cite{flash10}.

High performance implementations of cuckoo hashing in both serial
\cite{ross07} and parallel \cite{fan13, li14} have proven cuckoo
hashing to have practical potential. Typically, designs choose for
each bin to contain either four or eight slots.

Unfortunately, whereas cuckoo hashing reads are constant time
regardless of table-density, inserts become very slow for high-density
tables. The kick-out chain resulting from an insertion in a
high-density table results in the viewing of dozens (or hundreds) of
bins, making cache-friendly or multi-threaded cuckoo hashing difficult
\cite{li14}. Commonly kick-out victims are selected in one of two
manners: \emph{random kicking} selects the victim randomly; and
\emph{breadth-first search} performs a search to find the shortest
possible kick-out chain. Experimentally, at high table densities both
kick-out eviction algorithms are equally bad in terms of bins viewed
per insertion (Figure~\ref{fig_A}).

This short paper demonstrates that simple algorithmic changes to both
breadth-first search and random kicking can yield significant
improvements. We introduce five mechanisms for reducing bins viewed
per insertion: ghost insertions, sorted search, queue-kicking, and
rattle-kicking.

\begin{itemize}
\item \textbf{Ghost insertions: }Ghost insertions allow a record to
  reside in two bins at once, marked as a duplicate in each (Section
 ~\ref{secghost}). Duplicates make good kick-out victims because they
  are guaranteed not to cause additional kick-outs.

\item \textbf{Sorted Search: }Whereas breadth-first search essentially
  maintains a queue of records whose children have not yet been
  examined, sorted search instead maintains a list sorted by a
  statistic called \emph{spawn count} (Section
 ~\ref{sec_sorted_search}); records with small spawn counts are more
  likely to lead to small kick-out chains.

\item \textbf{Queue-Kicking:} To select kick-out victims,
  queue-kicking picks the record present in the bin for the longest
  (Section~\ref{sec_queue_walk}). For delete-light workloads, by
  preemptively updating the queue, queue-kicking can be further
  harnessed to prevent concurrently planned kick-out chains from
  overlapping.

\item \textbf{Rattle-Kicking: }In $d$-ary cuckoo hashing, records are
  hashed to $d$ single-slot bins \cite{dary}. Rattle-kicking ensures
  that each of a record's hash functions are used once before any of them
  are used a second time (Section~\ref{sec_dary_walk}).
\end{itemize}

Although we sometimes are able to provide theoretical justification
for our results, our primary method of evaluation is experimental. At
high table densities (using bins of size four) we obtain a factor of
ten improvement by combining sorted search and ghost insertions. At
high table densities for $4$-ary cuckoo hashing, rattle-kicking,
sorted search, and an algorithm of Khosla \cite{khosla13} each perform
approximately three times better than standard kick-out schemes.

In Section~\ref{secconcurrent}, we implement a transactional
multi-writer cuckoo hash table (without using hardware transactional
memory). For delete-light workloads, with 15 threads concurrently
running 100-operation transactions, queue-kicking reduces the
frequency of transaction aborts by a factor of around 256 in
low-density tables, and 32 in high-density tables (Section
\ref{secconcurrentexp}). To achieve similar results for arbitrary
workloads, we introduce a concurrency mechanism called a claim flag.
\begin{itemize}
  \item \textbf{Claim Flags: } By marking which slots have been
    scheduled for kick-out chains (and other write operations), claim
    flags eliminate all competition between inserts with high
    probability (even for polynomially many threads).
\end{itemize}
In our experimental evaluation for delete-heavy workloads, with 15
threads concurrently running 100-operation transactions, claim flags
reduce transaction-abort frequency by a factor of around 512 in
low-density tables, and 20 in high-density tables
(Section~\ref{secconcurrentexp}). Claim flags are a particularly
interesting direction for future work since they can likely be applied
to other transactional data structures.

The rest of this paper is organized as follows. Section
\ref{sec_past_work} discusses past work on fast cuckoo
insertions. Section~\ref{secexp} experimentally evaluates
ghost-insertions, sorted search, queue-kicking, and
rattle-kicking. Section~\ref{secghost} discusses ghost-insertions.
Section~\ref{sec_sorted_search} elaborates on sorted search, and
discusses its applications to concurrent cuckoo hashing. Section
\ref{sec_queue_walk} discusses queue-kicking, and proves for
delete-light workloads that a variant of queue-kicking eliminates all
competition between concurrent inserts with high probability. Section
\ref{sec_dary_walk} discusses rattle-kicking and compares it to an
algorithm of Khosla \cite{khosla13}. Section~\ref{secconcurrent}
implements an optimistic concurrency scheme for transactional cuckoo
hashing and introduces claim-flags; with high probability, claim-flags
eliminate all competition between inserts, regardless of
workload. Finally, Section~\ref{seccon} concludes with directions for
future work.

\section{Past work on fast cuckoo insertions} \label{sec_past_work}

Several techniques can be used to reduce kick-out chain length in
serial cuckoo hashing, the simplest of which would be to just use
large bins. One common technique is \emph{load balancing} in which,
upon inserting a record, one inserts into the least full of its two
bins. Another is \emph{stashing}, in which cuckoo chains longer than a
certain maximum length are aborted and the homeless key is inserted
into a separate stash \cite{kirsch09}. This helps especially in the
extremely rare event that no valid kick-out chain exists or if a
kick-out chain is abnormally long. Stashing was used, for example, in
the flash-based key-value store system FlashStore \cite{flash10}. One
of the cleverest modifications to cuckoo hashing is that of
\cite{lehman09}, in which rather than hashing keys to bins of size
$B$, one hashes keys to $B$ adjacent slots -- essentially allowing
bins to overlap. This simple improvement yields surprising performance
improvements. Although these techniques can be useful, to maintain
generality we evaluate our kick-out eviction algorithms in the
standard setting. In Section~\ref{secconcurrent}, however, we do take
advantage of load balancing to reduce contention in the multi-writer
setting.

Until now, not much work has been done on finding good kick-out
eviction schemes.
For $d$-ary cuckoo hashing (i.e., using bins of size one and $d$ hash
functions), however, algorithmic improvements to random-kicking have
been proposed by Khosla \cite{khosla13}, the performance of which we
compare to our results in Section~\ref{sec_queue_walk}.

\section{Experimental Evaluation of Kick-Out Eviction Algorithms}\label{secexp}

In the following sections, we will introduce a number of techniques
for reducing bins-viewed-per-insert in cuckoo hashing. In this section
we compare these techniques experimentally. Note that, in order to
facilitate comparison between algorithms, our graphs are on
logarithmic scales.

Figure~\ref{fig_A} compares random kicking, breadth-first-search,
sorted search (Section~\ref{sec_sorted_search}), and
queue-kicking\footnote{Our implementation of queue-kicking
  additionally uses hit-balancing, which is described in Section
 ~\ref{sec_queue_walk}.} (Section~\ref{sec_queue_walk}). For each
algorithm, we fill 1,000 hash tables to 97.5\% full and graph the
average number of bins viewed for insertions at each density. Each
table consists of $2^{13}$ bins, each of which has four slots. For
each algorithm, we also consider in Figure~\ref{fig_B} the version of
the algorithm in which ghost insertions have been implemented (Section
\ref{secghost}). Figure~\ref{fig_A} additionally tests the sorted
search / breadth-first search hybrid discussed in Section
\ref{sec_sorted_search}.

Some implementations of cuckoo hashing use a variant called
\emph{d-ary cuckoo hashing}. In this variant, bins are of size one
and, in order to still get good performance, $d$ hash functions are
used. Figure~\ref{fig_C} compares kick-out algorithms for 4-ary cuckoo
hashing, evaluating random-kicking, breadth-first search, sorted
search, rattle-kicking, and an algorithm due to Khosla
\cite{khosla13}, the final two of which are discussed in Section
\ref{sec_dary_walk}.  Just as in Figure~\ref{fig_A},
Figure~\ref{fig_C} reports the average number of bins viewed per
insertion over the course of 1000 trials; we use tables with $2^{13}$
single-slot bins.

Our experiments use uniformly randomly generated hashes. Moreover,
insertions assume that the key being inserted is not already present,
and are not responsible for verifying this themselves.

All of our experiments appear to scale. That is, if one runs the same
experiments on tables with arbitrarily many bins, then the data-points
will remain essentially unchanged.

\section{Ghost Insertions}\label{secghost}

In this section, we introduce the notion of ghost insertion, which can
be used to improve an any kick-out algorithm's performance for
high-density tables. When we insert a record $R_1$, we may have room
for it in both bins $b_1$ and $b_2$ to which it is hashed. Preferably,
we would insert $R_1$ into whichever of $b_1$ or $b_2$ there will be
the least demand for later. Ghost insertions simulate this by
temporarily inserting the record in both bins. Each copy of the record
is marked as a \emph{duplicate}, indicating to future insertions that
it can easily be removed to make space for another record.
  

Later we may find ourselves trying to insert a record $R_2$ into bin
$b_1$, only to discover that no slots are available. Fortunately,
because $R_1$ appears in both $b_1$ and $b_2$, we can simply remove it
from $b_1$. At this point, the copy of $R_1$ in $b_2$ is marked as no
longer being a duplicate record.

Surprisingly, all kick-out chains are guaranteed to terminate in a bin
containing a duplicate.

\begin{prop}\label{propghost}
  Let $C$ be a cuckoo-hash table built using only insert operations,
  and with ghost-insertions enabled. Let $b$ be the final bin in a
  kick-out chain in $C$. Then prior to the kick-out chain, $b$
  contained at least one duplicate.
\end{prop}
\begin{proof}
  Call a bin \emph{available} if it contains either a free slot or a
  duplicate. Call a bin \emph{reachable} if some record in $C$ is
  hashed to the bin but is contained (perhaps as a duplicate) in a
  different bin.

  It suffices to show that there does not exist a bin $b$ which is
  available, reachable, and free of duplicates. Suppose otherwise. As
  an available and duplicate-free bin, $b$ must contain at least one
  free slot. It follows that no record (duplicate or otherwise) has
  ever been kicked out of $b$. But since $b$ is reachable, some record
  $r$ is hashed to $b$ and contained in a different bin. However,
  because $b$ contained a free slot when $r$ was inserted, and since
  no record has ever been kicked out of $b$, record $r$ must be
  present in $b$ as a duplicate, a contradiction.
\end{proof}

Our experiments in Section~\ref{secexp} show that ghost insertions can
significantly reduce the number of bins viewed per insertion for each
kick-out eviction scheme. At high table densities, random kicking is
improved by a factor of roughly 2.5, and breadth-first search is
improved by a bit less than a factor of two. 

The factor of two is hinted at by Proposition~\ref{propghost}. At high
table densities, almost all available bins contain only one free or
duplicate slot. Proposition~\ref{propghost} tells us that only those
bins containing a duplicate are reachable. Thus ghost-insertions
essentially double the number of available reachable bins.

\section{Sorted Search} \label{sec_sorted_search}

In this section, we introduce sorted search, a kick-out
eviction algorithm based on breadth-first search. Experimentally, at
high table densities, sorted search can reduce the bins viewed per
insertion by a factor of eight (Figure~\ref{fig_A}).

We begin with a convention. A record's \emph{children} is the set
of slots which are in bins hashed to the record but not containing the
record. When conducting a search algorithm for a kick-out path, we use
the term \emph{spawning a record} to mean looking at the record's
children as part of the search. In turn, the record's children are its
\emph{spawn victims}.

Breadth-first search essentially maintains a queue of records that
have been viewed but not yet spawned in the search. At each step in
the search, we pop a record from the queue and check if any of its
children slots are free. If any are, then the search is
complete. Otherwise, we update the queue and continue.


Rather than maintaining a queue, sorted search maintains a list sorted
by some statistic. At each step, rather than picking the record which
has been in the list the longest to spawn next (as in breadth-first
search), we pick the smallest record according to the statistic.

After experimenting with many statistics, we have found one in
particular, which we call spawn count, to be the most effective
by far. The \emph{spawn count} of a record in bin $b$ is the number of times
since the conception of the hash table that any search has previously
spawned any record that was, at the time, contained in
$b$. Surprisingly, even statistics using information specific to the
record (and not just the bin) appear unable to perform better than
spawn count.


For delete-heavy workloads, spawn count could potentially get
large. One might choose to cap it at a given value. In our
(delete-free) experiments, four bits (per bin) is easily sufficient to
store spawn count. Consequently bucket sort can be used to efficiently
maintain the sorted list during each search.


Unlike breadth-first search, which guarantees to return the shortest
possible cuckoo path, sorted search could potentially return a path
containing a cycle -- that is, we may accidentally try to kick-out the
same record from the same bin twice in the same kick-out chain. One
low-overhead method to avoid revisiting buckets is to maintain a small
hash table of the already visited buckets; in particular, the index of
a bucket can be used as its hash. In order to fairly compare
breadth-first search and sorted search, we eliminate the revisiting of
buckets in our implementations of both.

Sorted search maintains two significant concurrency advantages of
breadth-first search over random kicking. It can be implemented to
take advantage of prefetching, and it provides short kick-out
chains.


\textbf{Prefetching:} Unlike random walking, breadth-first search
tells us the next few bins that will be fetched before we actually
need to fetch them. In turn, this facilitates the use of prefetching to
reduce fetch latency \cite{li14}.

Sorted search can also take advantage of prefetching. Indeed, sorted
search may be modified to spawn the $k$ smallest-statistic records at
a time rather than just one (for an arbitrary $k$). Moreover, if more
than two hash functions are being used, one can take advantage of
prefetching even when just fetching the children of a single record.

\textbf{Short Critical Path:} Breadth-first search yields a kick-out
chain whose size is logarithmic in the number of bins visited
during the search. Previously, \cite{li14} utilized this to
significantly reduce critical path size in concurrent cuckoo hashing.

Experimentally, sorted search also satisfies this property. In
Figure~\ref{fig_D}, we compare the length of kick-out chains generated
by sorted search and breadth-first search both with and without
ghost-insertions (using the same experimental set-up as in
Figure~\ref{fig_A} from Section~\ref{secexp}). As can be seen by
comparing Figures~\ref{fig_A},~\ref{fig_B}, and~\ref{fig_C}, at 97.5\%
table density, all four search-based algorithms produce kick-out
chains an order of magnitude smaller than those produced by walk-based
algorithms.

If we model "picking the smallest-statistic bin" as "picking a random
bin," and assume that each spawn is equally likely to terminate the
search, then we can prove that sorted search's kick-out chains are
logarithmic in comparison to the number of bins viewed.

\begin{thm}
Suppose we search for a kick-out chain by, at each step in the search,
spawning a record selected randomly from those viewed but not yet
spawned. Moreover, suppose that each spawn is equally likely to
terminate the search.
  
Then the expected length of the resulting kick-out chain is $O(\lg
r)$, where $r$ is the total number of spawns in the search.
\end{thm}
\begin{proof}
Let $B$ be the size of each bin and $H$ be the number of hash
functions. We require that $B(H-1) > 1$. Then for $k \ge 0$,g at step
$k$ we have viewed $2B + kB(H-1)$ records. Since each step (after
zero) spawns one record, that leaves $2B + kB(H-1)-k$ records which
have been viewed but not spawned. Let $f(k)$ be the average
search-tree depth of the records so far viewed but not yet spawned
after $k$ steps, with $f(0) = 0$. The $k$-th step in the algorithm
eliminates one record with expected depth $f(k-1)$ and introduces
$B(H-1)$ records with expected depth $f(k-1)+1$ each. Since after the
$k$-th step there are $2B + kB(H-1)-k$ unspawned but viewed records,
we can treat all but $B(H-1)$ of those records as having average depth
$f(k-1)$ and $B(H-1)$ as having average depth $f(k-1)+1$, yielding
total average depth of $$f(k) = f(k-1) + \dfrac{B(H-1)}{2B + kB(H-1) -
  k}$$ $$\le f(k-1) + \dfrac{1}{k - k/(B(H - 1))}.$$ Since $B(H-1)
\geq 2$, it follows that $f(k) \leq 2 (1+1/2+1/3+\cdots 1/k) \approx
2\ln k$. Thus in a search with $r$ spawns, the expected depth of the
final spawn is $O(\lg r)$.
\end{proof}

One may wish to use a sorted search and breadth-first search
hybrid. Here, a record's position in the sorted list is determined
first by the depth in the search at which it was a spawn victim, and
then secondarily by its spawn count. This erases the necessity of
monitoring for cycles and guarantees the kick-out path-length being
logarithmic in bins viewed. As expected, the hybrid's performance
is experimentally between that of sorted search and breadth-first
search (Section~\ref{secexp}).



\section{Queue-Kicking} \label{sec_queue_walk} 

Random walking is easy to implement and extremely low
overhead. Unfortunately, at high table densities it does poorly at
finding short kick-out chains. In this section, we introduce
queue-kicking, which improves upon random kicking in order to
significantly reduce bins viewed per insertion.

In order to develop a variant of random walking which will find
shorter kick-out chains, we must first understand why random walking
does poorly. Suppose one starts with an empty hash table (using bins
of size four), and then fills the table to 97\% full. Then the total
number of kick-outs during all of the insertions combined will only be
around the half the total number of records
inserted\footnote{Experimentally, regardless of table size, each bin
  tends to average 2.07 kick-outs.}

But if we pick kick-outs within a bin randomly, balls-in-bins
suggests that some slots will get picked several times while others
won't get picked at all. Consequently, we will end up kicking records
back into a bin they previously got kicked out of, when we instead
could have kicked out a record who had never been kicked around
before.

A simple resolution is to implement each bin as a queue. Rather then
selecting kick-out victims randomly, one simply selects the record
which has been in the bin the longest. We call this
\emph{queue-kicking}. Queue-kicking's performance is additionally
bolstered by the fact that the earlier a record is inserted into a
bin, the less likely it is there only because the
other bin to which it was hashed was full.

Our experiments in Section~\ref{secexp} show that experimentally
queue-kicking results in far shorter kick-out chains than does
random-walking, at least for high-density tables. We may not be the
first to observe the benefits of queue-kicking; in fact, it is the
kick-out scheme used by Kennith Ross in his high-performance cuckoo
hash table implementation \cite{ross07}. We are the first, however, to
explicitly observe its benefits -- previously it appeared only in a
single clause of a single sentence.

For delete-light loads, one could simulate queue-kicking by using a
counter for each bin. This counter, known as the \emph{hit-counter},
is incremented each time a record is inserted into a bin. The record
is then placed in the slot whose index is congruent to the hit-counter
modulo the size of the bin, kicking out another record if
needed.\footnote{In our experiments the hit-counter can easily be
  stored within a single byte (similarly to spawn-count for sorted
  search). One could choose to store the hit-counter modulo the number
  of slots in each bin.}

Using hit-counters would not work well for delete-heavy workloads
since they could result in a kick-out in a bin which has free
slots. For a workload not containing many deletes, however,
hit-counters have two surprising consequences, the first a small
benefit for performance, and the second a major benefit for concurrency.

\textbf{Hit Balancing: } Previously we discussed a common technique
called \emph{load balancing} in which, upon inserting a record, one
views both hashed bins and picks the less full one. When both
bins are full, however, we can generalize load balancing to
\emph{hit balancing} by picking the bin with a smaller
hit-counter. So that they can utilize hit balancing, our experiments
in Section~\ref{secexp} use hit counters to implement queue-kicking.

\textbf{Scheduling Kick-Out Chains: } One problem with long kick-out
chains is that, in multi-threaded systems, two concurrent insertions
may plan overlapping chains, forcing one of the insertions to start
over \cite{li14}. As we will see in Section ~\ref{secconcurrent}, this
can lead to a problematic number of transaction aborts in the context of
transactional cuckoo hashing. In particular, in transactional cuckoo
hashing, hundreds or thousands of inserts may be planned before any of
them are committed to the table. Any overlaps in their kick-out chains
will lead to transaction aborts. Worse still, if two kick-out chains
terminate in the same bin which contains only a single free slot, they
are guaranteed to compete for the slot.

  Interestingly, this problem can be resolved using hit-counters. If
  hit-counters are incremented with atomic fetch-and-adds, then we are
  guaranteed that any $B$ threads trying to insert into the same bin
  will all be assigned different slots. (Here, $B$ is the number of
  slots in the bin.) In practice, for delete-light workloads, this is
  enough to reduce transaction aborts by a factor of 30 (Section
 ~\ref{secconcurrent}).


In fact, we can prove with high probability that the use of
hit-counters completely eliminates overlaps between concurrent
inserts. This is formalized through the following
lemma.

\begin{lem}\label{lemhighprob}
  Consider a cuckoo hash table with $n$ bins, each containing $B$
  slots. Suppose $t$ threads concurrently each perform $D$ operations,
  with each operation randomly touching at most $j$ slots. If $D \in
  O(n)$, then with probability $O((tj)^{B + 1}n^{1-B})$ no bin will
  ever simultaneously receive more than $B$ touches.
\end{lem}
\begin{proof}
  Observe that at most $t \cdot j$ touches can occur
  simultaneously. Consider an arbitrary bin $b$. Out of the $n^{t\cdot
    j}$ ways to assign $t \cdot j$ distinct touches to $n$ bins, at
  most $${{t \cdot j} \choose {B + 1}} n^{t \cdot j - (B + 1)}$$ of
  them assign more than $B$ touches to bin $b$. By the union bound,
  the probability of any bin being over-subscribed is at most
  $$\frac{{{t \cdot j} \choose {B + 1}}}{n^B} \le t \cdot j \cdot
  \left(\frac{t \cdot j}{n}\right)^B.$$ Applying the union bound to
  the $D$ operations, the probability of any bin is ever
  over-subscribed is at most
  $$D \cdot t \cdot j \cdot \left( \frac{t \cdot j}{n} \right)^B \in
  O((tj)^{B + 1}n^{1-B}).$$
\end{proof}

\section{Rattle-Kicking} \label{sec_dary_walk} 

Until now, we have focused on cuckoo hash tables using two hash
functions and bins with several slots. Some applications, however, use
\emph{$d$-ary cuckoo hashing}, in which bins contain one slot only and
$d$ hash functions are used \cite{dary}.\footnote{Two examples of
  systems using this technique include Shore-MT \cite{flash10} and
  FlashStore \cite{johnson09}.}

For $d$-ary cuckoo hashing, random kicking performs poorly with
respect to bins viewed per insertion, even for tables at smaller
densities. In this section, we resolve this with a simple technique
called rattle-kicking.

Refer to the $d$ hash functions as $h_0$, $h_2$, ..., $h_{d-1}$. For
each key $x$ in the hash table, we maintain a \emph{rattle-counter}
denoted by $r(x)$ which is initially zero and is incremented every
time we try to insert $x$ into a bin. When a key $x$ is either
introduced to the table or displaced from its previous bin, we try to
insert it into bin $h_{r(x) \pmod d}$. If that bin already contains a
key $y$, however, then whichever of $x$ or $y$ has a higher
rattle-counter gets to stay in the bin. We increment the other's
rattle-counter and recursively try to insert it into another bin.


We refer to the act of picking the key with the smaller rattle-counter
as \emph{rattle-balancing}. Similarly to load balancing,
rattle-balancing ensures that there is little variance between
rattle-counters of records throughout the hash table.

At the cost of a small overhead (rattle-counters) per key, our
algorithm appears to significantly reduce the length of kick-out
chains. (See Section~\ref{secexp}.) At table density below 85\%, our
algorithm performs close to as best as one could hope. If a table has
density $\delta$, then we would hope each insert and kick-out would
have a $(1-\delta)$ chance of not inducing another kick-out. In turn,
this would yield chains of average length $1/(1-\delta)$.

Random kicking performs poorly because it reuses hash functions
unnecessarily.  According to our data for $4$-ary cuckoo hashing, when
we fill a table from empty to 95\% using random kicking, each record
has visited a total of only around 5.6 bins\footnote{This average
  appears to be mostly independent of table size.}. But if we stick
5.6 balls randomly in four cups (the number of hash functions), some
cups will get several balls while others may be totally empty. When we
try a hash function we've never used before for a key, we are
effectively picking a random bucket in the hash table in which to
insert that record -- this will lead to the $1/(1-\delta)$ performance that
we desire. But when we instead reuse a hash function, we are sending
the record to a slot we already know is full.

Rattle counters solve this problem by going through all the hash
functions for a record before reusing any. Moreover, rattle-balancing
keeps rattle-counters for all records relatively balanced, rather than
some records having many unused hash functions while others have
none. Using rattle-counters to fill a table to 95\% full results in
records having used only 3.3 of their hash functions on
average. Consequently, most records have not reused a single hash
function.


In addition to testing random-walking, rattle-kicking, breadth-first
search, and sorted search for d-ary cuckoo hashing (Section
\ref{secexp}), we test an interesting kick-out eviction algorithm due
to Khosla \cite{khosla13}. For high-density tables Khosla's algorithm
performs similarly to rattle-kicking.

The computations in Section~\ref{secexp} assume a record's prior
absence for each insertion. Consequently, when the table is near
empty, our algorithms view fewer than $d$ slots per
insertion. Khosla's algorithm must maintain certain
invariants which prevent it from easily doing this  \cite{khosla13}.

Rattle-kicking can be thought of as a $d$-ary analogue for
queue-kicking. Is there a natural analogue for Khosla's algorithm as well?

\section{Transactional Multi-Writer Cuckoo Hashing}\label{secconcurrent} 

In sections~\ref{sec_sorted_search} and~\ref{sec_queue_walk}, we
discussed the relation of several of our kick-out algorithms to the
multi-writer setting. In this section we apply queue-kicking to the
multi-writer setting, and introduce a novel concurrency scheme which
allows a multi-writer Cuckoo hash table to achieve transactional
correctness without being inhibited by transaction aborts. Past
authors have proven cuckoo hashing's potential as a multi-writer table
with threads performing single operations at a time \cite{fan13,
  li14}. Our work extends this to transactional Cuckoo hashing, where
each thread wants transactions comprising many operations to be
performed atomically.

Each thread in a multi-writer hash table may wish to perform a series
of transactions, where each transaction comprises a collection of
reads, writes, deletes, and overwrites which are dependent on each
other. For example, if a thread were managing bank data, it might wish
to move \$1000 from Ann's account to Bill's account under the
condition that Bill's account contains less than \$1000 and his friend
Liz's account contains less than \$500. Transactional correctness is a
property which prevents inter-dependencies between concurrent
transactions from producing unexpected results. A table is said to be
\emph{transactional} if after all the transactions have been performed
in parallel, there exists a serial schedule for the transactions which
would have resulted in the same end-result hash table.



Past non-transactional multi-threaded Cuckoo hash tables have avoided
locking slots for read operations by using version counters on each
slot, and having reads check that the counter did not change between
the start and end of the read \cite{fan13, li14}. In order to achieve
transactional correctness, we use a more complicated notion of version
IDing based on an optimistic-concurrently scheme of SILO, a
high-performance in-memory database \cite{silo13}. In particular, all
of the slots in a transaction are assigned a new version ID at the end
of the transaction; and the final hash table will be the same as the
one obtained from any serial scheduling of the transactions weakly
ordered by version-ID number.

Our experiments find that a naive adaption of SILO's optimistic
concurrency scheme to Cuckoo hashing results in frequent transaction
aborts, even at low table densities. We introduce a number of
techniques for reducing these aborts, including an application of
queue-kicking (introduced in Section \ref{sec_queue_walk}), local
retries, and claim flags. Using fifteen threads to concurrently build
a table, our techniques can be combined to reduce abort-frequency by a
factor of more than 7,000 for low-density tables and 450 for
high-density tables (Figure \ref{fig_F}). In particular, with high
probability, claim flags allow threads to schedule inserts and
kick-out chains without competing. In our experiments for delete-heavy
loads, claim flags bring the percent of transaction attempts that
abort down from more than 5\% to approximately .02\% at low table
densities, and from more than 10\% to approximately .5\% at high table
densities (Subsection \ref{secconcurrentexp}). These percentages will
vary depending on the size of the table, the number of threads, and
the number of operations per transaction. For example, a larger table
with the same number of threads and number of operations per
transaction will yield smaller percentages.

\subsection{A Naive Optimistic Concurrency Scheme}

We implemented the following transactional multi-writer cuckoo hash
table, using an optimistic concurrency scheme based on the techniques
of in-memory database SILO \cite{silo13}.

Each slot and each bin has its own spin lock and version ID number,
the latter of which guards the former. Each transaction has three
stages: \\ \textbf{Stage 1: The Planning Stage. }The thread performs a
transaction comprising inserts, look-ups, overwrites, and deletes;
each of which may depend on the outcome of previous operations. In
this stage, however, the transaction does not apply its edits to the
table. In order to achieve transactional correctness, the edits are
applied in Stage 3.

Prior to reading a slot's contents, the slot's version ID is added to
the \emph{transaction's read set}. Whenever the transaction plans to
edit a slot, however, the slot's version ID is instead added to the
\emph{transaction's write set}\footnote{After reading a slot's contents, the
transaction should check that the version ID is unchanged, in order to
guarantee the integrity of the data read.}. In addition, the transaction
will sometimes wish to verify that a record $r$ is not in a bin
$b$. Instead of using the version ID of each of $b$'s slots, the
transaction inserts $b$'s version ID into the read set. In turn,
insertions add to the write set the version ID of the bin into which
they plan to insert. \\ \textbf{Stage 2: The locking and verification
  stage. } The transaction locks the slots and bins whose version IDs
are in the write set\footnote{This must be done in a globally sorted
  order to avoid deadlock.}. The transaction then checks whether any
version ID in either the read set or write set has changed since being
added to the set. If any have, then the transaction releases all of
its locks and aborts. Version IDs in the write set can be checked
immediately after the corresponding slot/bin has been
locked. \\ \textbf{Stage 3: The apply/commit stage. } The transaction
applies each of the write-operations from Stage 1. Additionally, prior
to releasing the locks, the transaction updates the version ID of
every slot and bin in the write set to a new version ID, the
\emph{transaction ID}. The transaction ID is the maximum of all the
IDs in the read set along with all the IDs in the write set, except
additionally incremented by one\footnote{Additionally, we may require
  that it is at least one greater than the transaction IDs of any
  previous transactions run by the same thread.}.

Observe that each transaction effectively occurs atomically, in the
instant when the final lock in Stage 2 is taken.  Moreover, no other
transaction with a smaller transaction ID can possibly have made edits
depending on those of the current transaction, a property which can
prove useful for logging and snapshots.

Surprisingly, even at low table-densities transaction aborts can pose
a major obstacle. More than 40\% of transaction attempts abort, for
example, when 15 threads concurrently perform 100-operation
transactions to fill a $2^{17}$ slot table to 60\% full using a
delete-heavy workload. (See Subsection \ref{secconcurrentexp} for more
details.) Most of these aborts are overly conservative, and can be
eliminated through the use of local retries, which we will introduce
in the next section.

Even with local retries, however, more than 5\% of attempted
transactions abort, and for a table filled to 95\% full, more than
10\% of attempted transactions abort. In the next section, we will
eliminate almost all of these aborts using a new mechanism called a
claim flag. Note that these percentages are specific to the parameters
of our experiment, and would shrink, for example, if we increased the
table size without changing the number of threads or size of
transactions.


\subsection{Mechanisms for Reducing Aborts}\label{secmech}

The concurrency scheme described so far falls apart for high-density
tables due to an excess of transaction aborts. In addition to
queue-kicking (which was discussed in Section~\ref{sec_queue_walk}),
we experimentally evaluate three mechanisms for reducing transaction
aborts: local retries, system-transaction kick-outs, and claim flags.

\textbf{Local Retries: } Local retries eliminate a class of aborts
which are overly conservative. In Stage 1, the success of an
insertion, or the failure of a read, overwrite, or delete is
contingent upon a particular record being absent from the table. In
Stage 2, the continued absence of this record is confirmed using bin
IDs in the read and write sets. If one of these IDs has changed,
however, then the entire transaction aborts. Most likely, however, the
bin-transaction-ID change was not a result of the relevant record
appearing in the table. Thus when a thread $T$ finds during Stage 2
that a bin version ID used to confirm the absence of a record $k$ has
changed, thread $T$'s transaction should not abort. Instead it should
release its locks, verify that $k$ is still absent from the bin, and
restart Stage 2 using the bin's new version ID. This is referred to as
a \emph{local retry}.

Local retries for bin IDs in the write set can be performed without
releasing already acquired locks. In particular, by locking bins
before slots in Stage 2, one can perform local retries (without risk
of deadlock) on write-set bin IDs right before locking the
corresponding bin.


\textbf{System-Transaction Kick-Outs: } Kick-out chains need not wait
until Stage 3 to be performed. Instead, they can be performed as
individual transactions during Stage 1. In fact, by performing a
kick-out chain in reverse order, each individual kick-out can be
performed as a mini-transaction.

Surprisingly, without the aid of claim flags or queue-kicking,
system-transaction kick-outs would actually \emph{increase} the number
of kick-out chain collisions. Competition for the slot freed by the
kick-out chain leads to aborts.

\textbf{Claim Flags: } In stage 1, whenever a thread adds a slot to
the write set, it atomically \emph{claims} the slot using a one-bit
flag. Claims do not restrict read-access to a slot or to its version
ID.

Whenever an insertion or kick-out chain comes across a claimed slot,
it simply ignores that slot. If all of the slots in a bin are claimed
or locked, then the insertion / kick-out chain is forced to
abort. Lemma~\ref{lemhighprob} tells us with high probability,
however, this will never happen within a $n$-bin table's entire
$O(n)$-operation lifespan.

In order to avoid deadlock, no operation can ever wait on a claim
flag. Instead, when overwrites, or deletes see a claimed slot,
they must abort the transaction. In the same situation without claim
flags, however, one of the conflicting transactions would have been
forced to abort later regardless of the concurrency scheme.

Once a claim flag is taken, the slot's version ID is guaranteed not to
change before the transaction's completion. Consequently, slots need
not be locked until Stage 4, and slot locks need not be acquired
atomically. 

Besides the occurrence that a bin is completely claimed / locked
(which Lemma~\ref{lemhighprob} tells us is rare), claim flags and
local retries leave only four possible causes for aborts: a slot is read by one transaction but then modified by another;
 two deletes / overwrites try to edit the same record concurrently; a
 delete / overwrite tries to edit a slot already claimed for a
 kick-out chain; or a transaction inserts a key which another transaction
 is relying on as not being present.

\subsection{Experimental Results on Multi-Threaded Cuckoo Hashing}\label{secconcurrentexp}

We gather data on six variants of transactional multi-threaded cuckoo
hashing. The first uses neither local retries, claiming, system
transaction kick-out chains, nor queue-kicking. We significantly
improve this with the second, which adds in local retries. In the
third, we introduce queue-kicking (using hit-counters). For a
delete-light workload this both decimates transaction aborts and
shortens kick-out chains. In the fourth implementation, we introduce
system-transaction kick-out chains, achieving additional
improvements. The fifth implementation replaces queue-kicking and
system-transaction kick-outs with claim flags, resulting good
performance regardless of workload. Finally, the sixth implementation
introduces system-transaction kick-out chains again\footnote{Note that
 unlike for queue-kicking, abort reductions from system-transaction
  kick-out chains are not already obtained by claim flags.}.

We run tests on a table with $2^{14}$ bins, each containing $8$
slots. We consider two workloads, one of which is light on deletes and
one of which is delete heavy; the relative frequencies of inserts,
deletes, overwrites, and reads is 1:0:1:1 and 2:1:2:2 in the two
tests. Our experiments use 15 threads, each of which runs transactions
batching together 100 operations. Each insert uses a randomly
generated pair of hashes and a randomly generated integer payload;
insertions use load balancing, but not hit balancing. To randomly
select a record for an overwrite, delete, or read, a thread picks out
of all records ever inserted into the table by the thread.

Figure~\ref{fig_E} shows the results for the delete-light load, and
Figure~\ref{fig_F} shows the results for the delete-heavy load. In
each case, we examine the total number of aborts for a table filled to
each density, and average our results over 100 trials. Note that both
graphs are on a log scale.

In general, one expects the total number of aborts to grow
proportionally to the number of threads and number of operations per
transaction, but to not depend on table size. Running our experiments
on tables with varying parameters indicates that this is roughly the
case, but it would be interesting to study the scaling properties of
each implementation in more detail.


\section{Directions for Future Work}\label{seccon}

Several directions for future work present themselves.

Our results are mostly experimental. Can one theoretically quantify
the improvements obtained by our algorithms? Ghost insertions seem
particularly interesting to study.

What other data structures can claim flags be applied to in order to
achieve transactional correctness with low risk of transaction aborts?
Is their applicability widespread?

Our tests use uniformly random hash functions. How do the relative
performances of the algorithms change when bias hash functions are
introduced?

Our tests compare kick-out eviction schemes by bins viewed per
insertion, and concurrency schemes by transaction-abort frequency. It
would be interesting to also compare runtime performances in various
settings.

\section{Acknowledgments}\label{secawk}

This research was conducted at HP Labs. The author thanks Harumi Kuno,
Hideaki Kimura, and Bradley Kuszmaul for offering advice on exposition
and directions of research. The author additionally thanks Hideaki
Kimura for suggesting the notion of system-transaction kick-out
chains.

\newpage
\bibliographystyle{abbrv}
\bibliography{cuckoobib} 
\clearpage
\clearpage

\section{Figures}

\begin{figure}[!htb]
  \includegraphics[scale=.3]{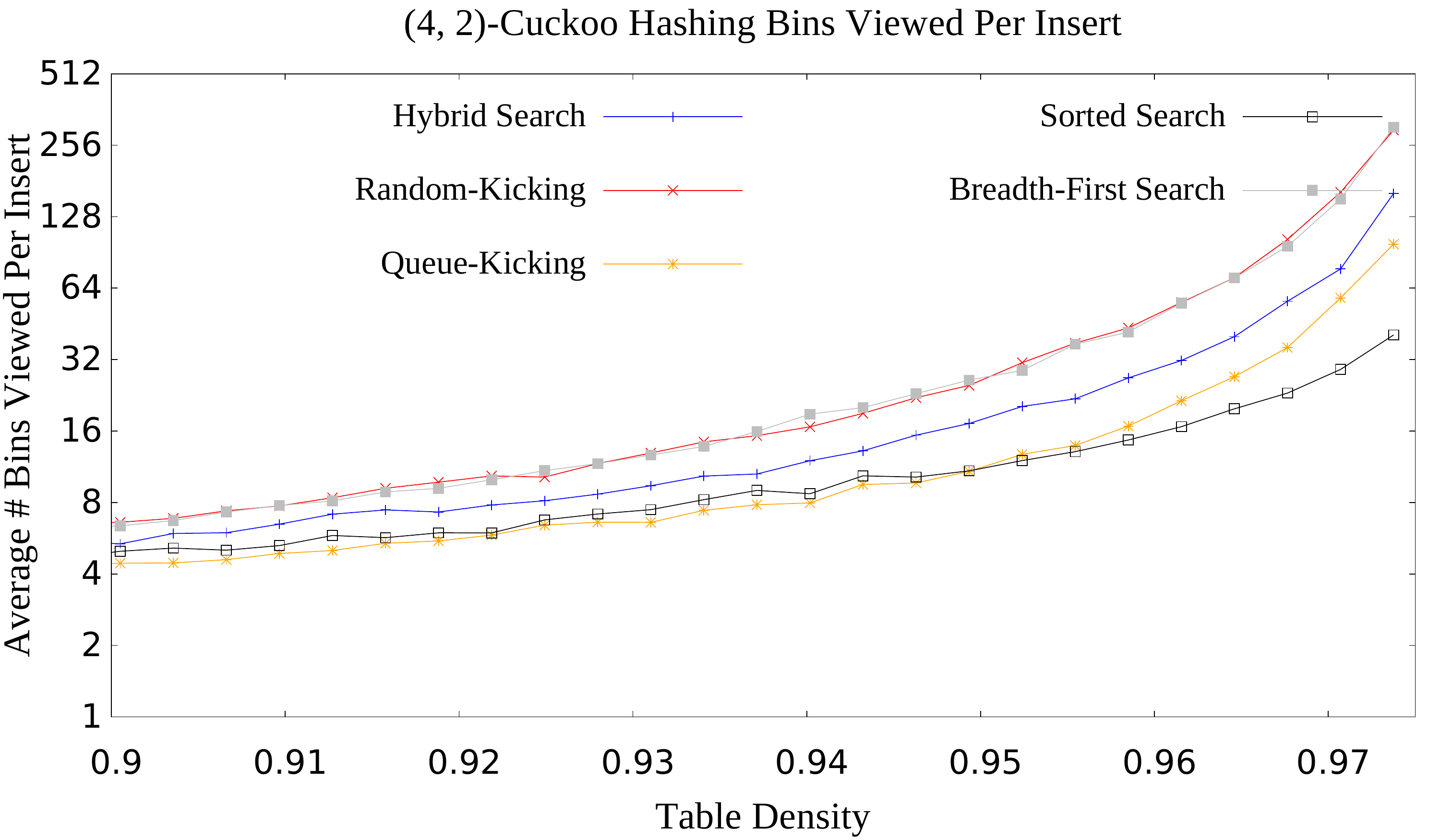}
  \caption{The average number of bins viewed per insertion at varying
    table densities, for various kick-out algorithms.}
  \label{fig_A}
\end{figure}

\begin{figure}[!htb]
  \includegraphics[scale=.3]{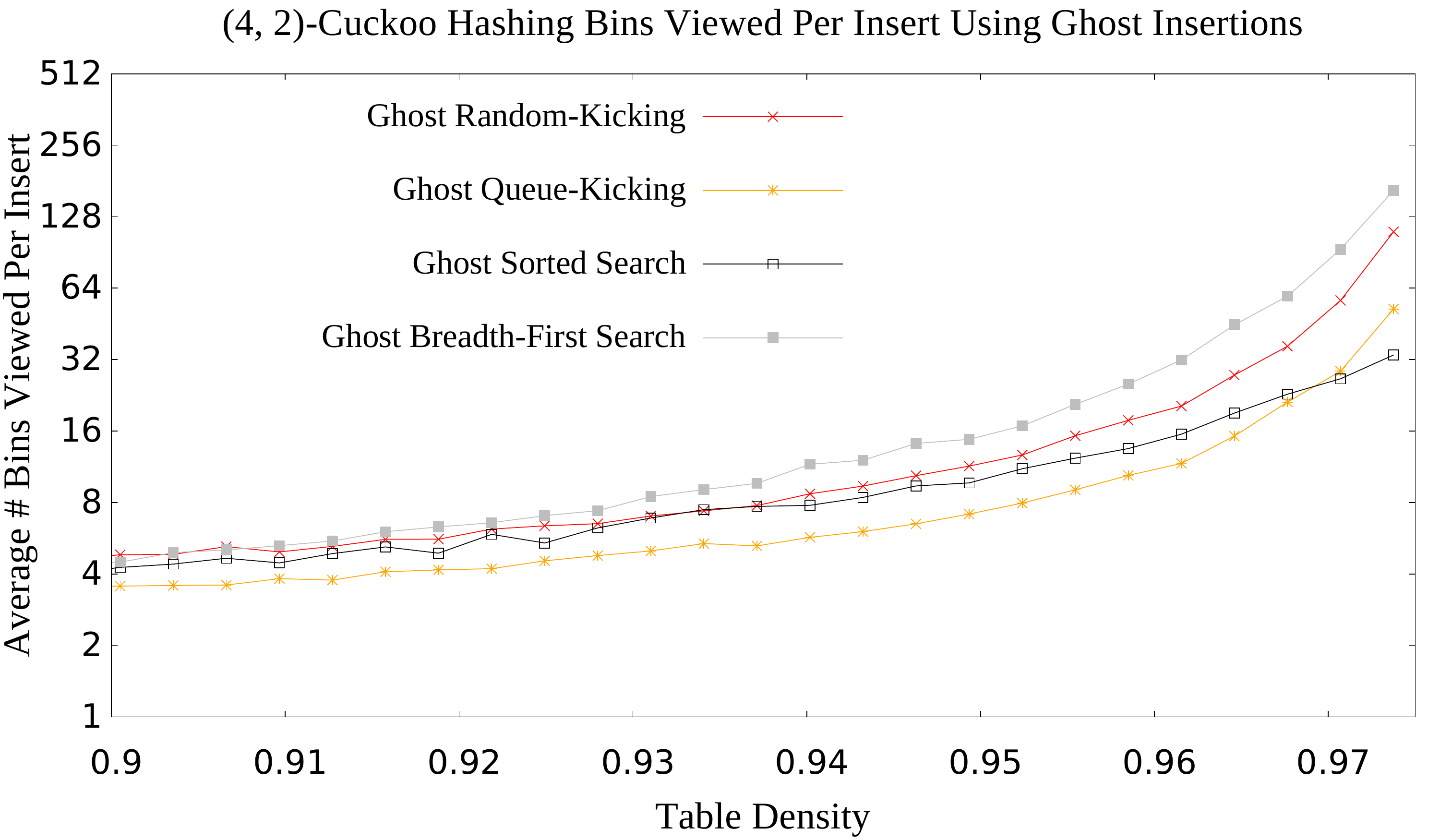}
    \caption{The average number of bins viewed per insertion at
      varying table densities, for various kick-out algorithms using
      ghost-insertions.}
  \label{fig_B}
\end{figure}

\begin{figure}[!htb]
  \includegraphics[scale=.3]{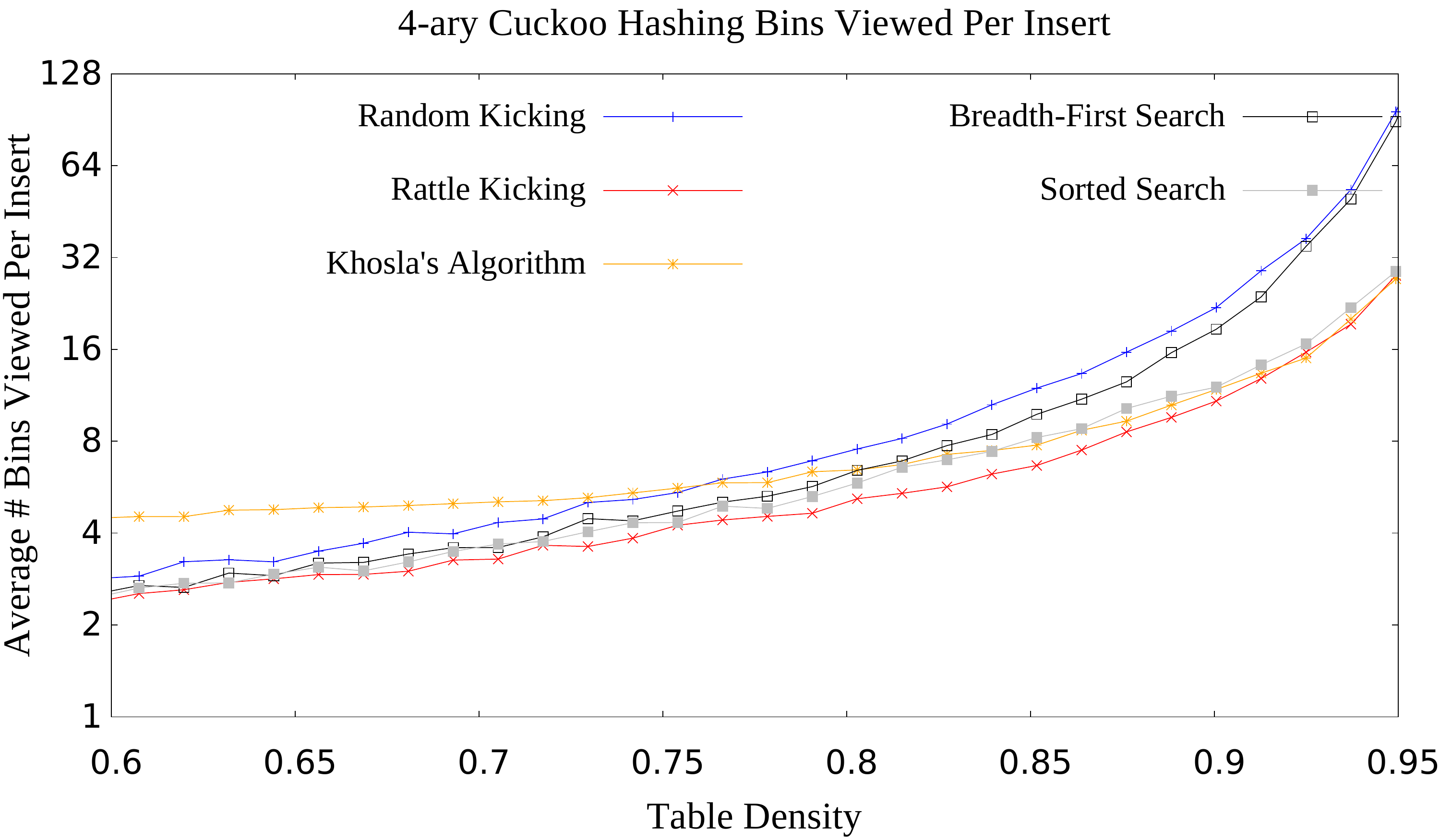}
      \caption{The average number of bins viewed per insertion at
        varying table densities, for various $4$-ary kick-out
        algorithms.}
  \label{fig_C}
\end{figure}

\begin{figure}[!htb]
  \includegraphics[scale=.3]{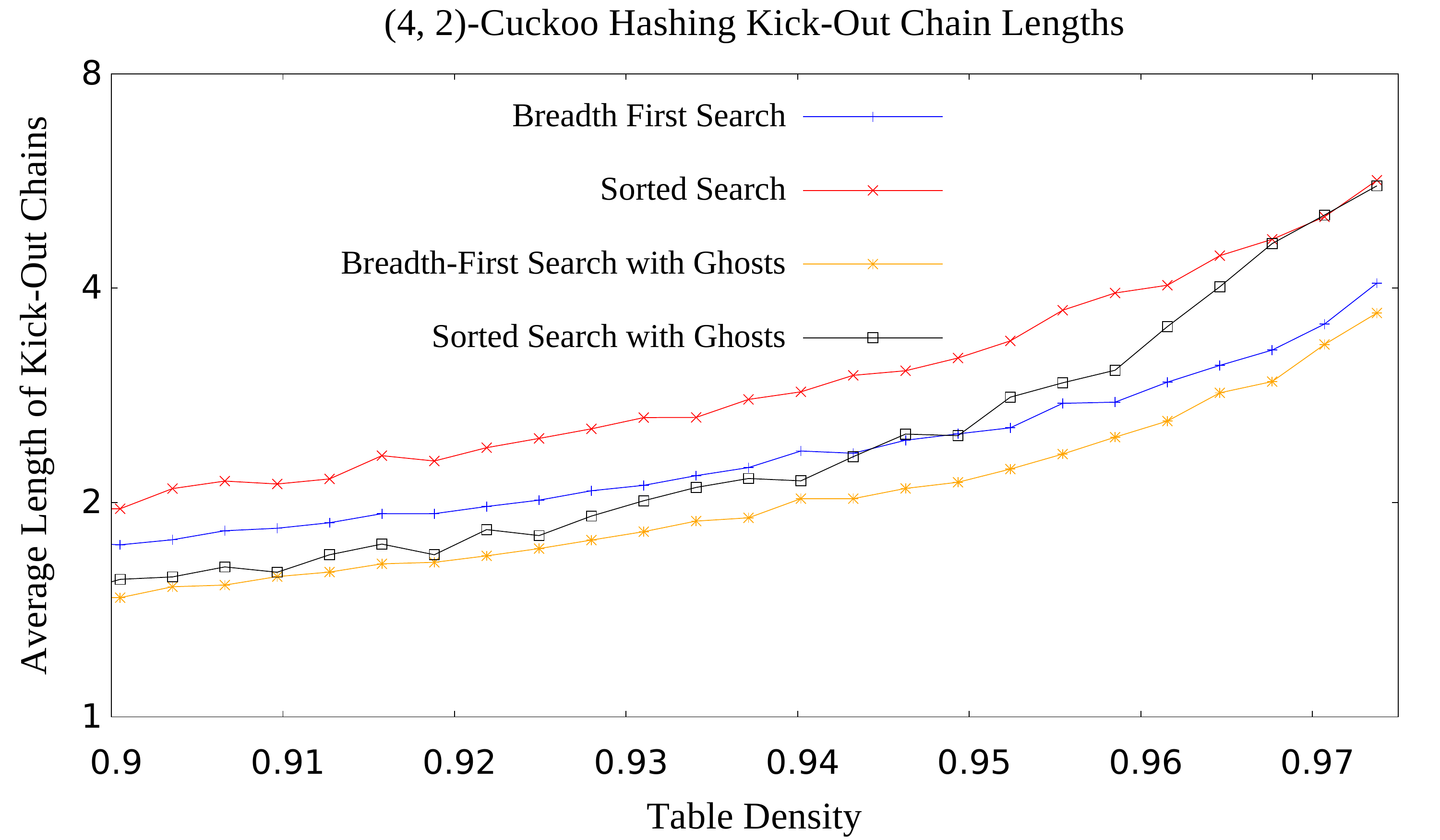}
  \caption{The average length of kick-out chains at varying table
    densities, for various search-based kick-out algorithms.}
  \label{fig_D}
\end{figure}

\begin{figure}[!htb]
  \includegraphics[scale=.3]{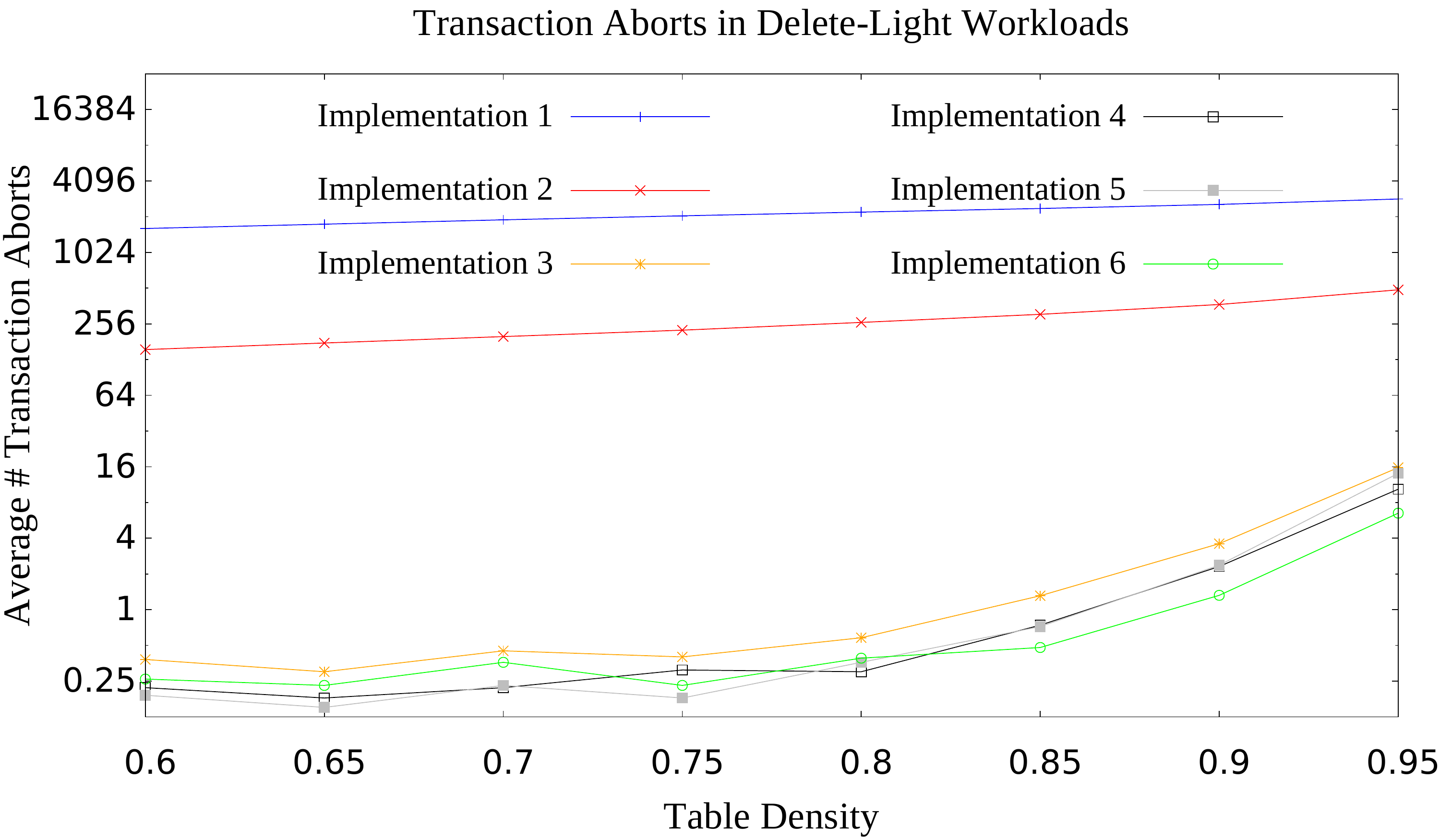}
  \caption{The average number of transaction aborts over the lifetime
    of tables filled to varying densities, for various transactional
    cuckoo hashing implementations, using a delete-light workload and
    15 threads.}
  \label{fig_E}
\end{figure}

\begin{figure}
  \includegraphics[scale=.3]{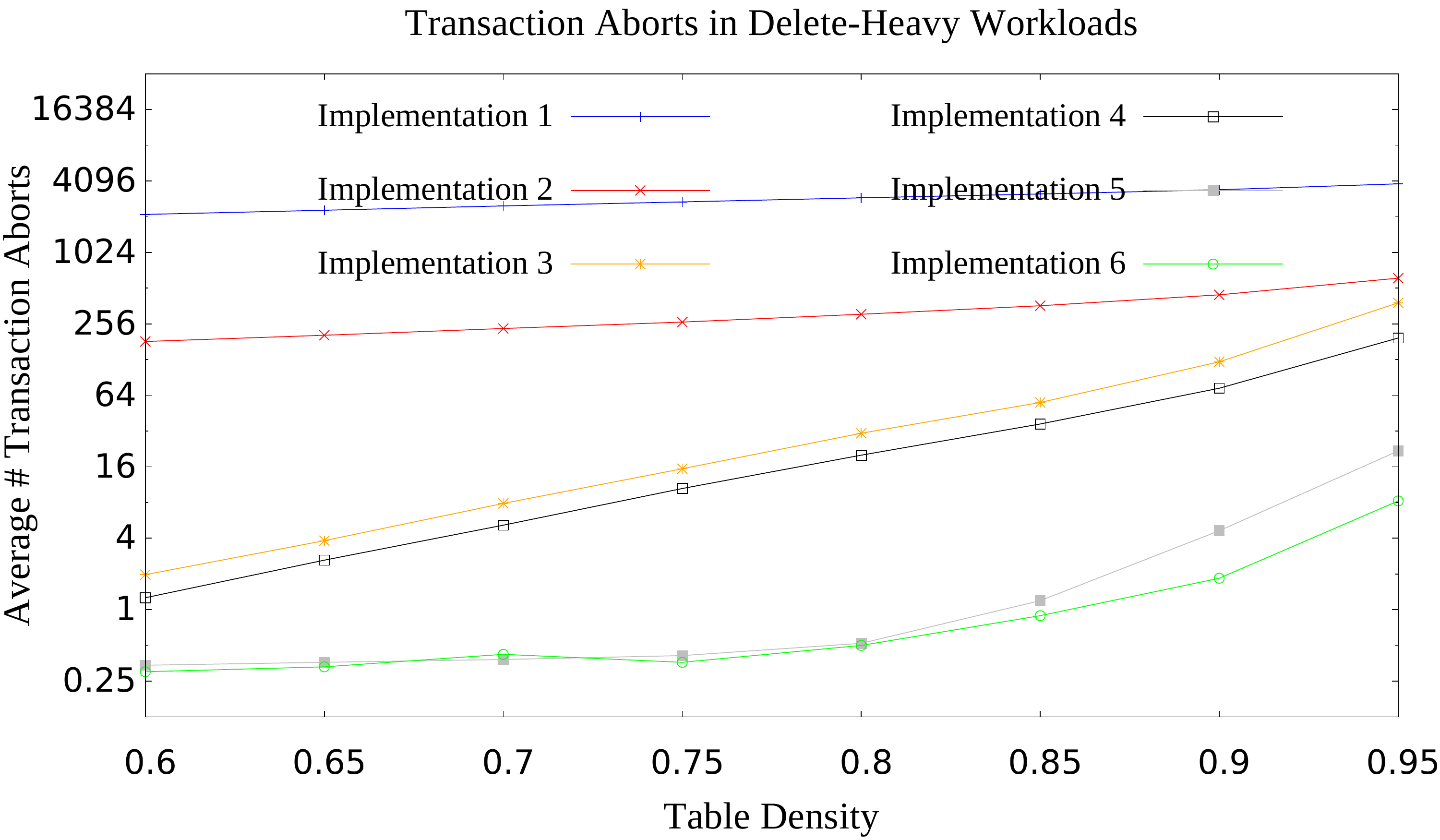}
  \caption{The average number of transaction aborts over the lifetime
    of tables filled to varying densities, for various transactional
    cuckoo hashing implementations, using a delete-heavy workload and
    15 threads.}
  \label{fig_F}
\end{figure}

\end{document}